\documentclass[letterpaper, 10 pt, conference]{ieeeconf}

\IEEEoverridecommandlockouts
\usepackage{graphicx} 
\usepackage{xcolor}
\usepackage{amsmath}
\usepackage{amssymb}
\usepackage{algorithm}
\usepackage{algpseudocode}
\usepackage{mathabx}
\usepackage{wrapfig}

\makeatletter
\newcommand\fs@betterruled{%
  \def\@fs@cfont{\bfseries}\let\@fs@capt\floatc@ruled
  \def\@fs@pre{\vspace*{5pt}\hrule height.8pt depth0pt \kern2pt}%
  \def\@fs@post{\kern2pt\hrule\relax}%
  \def\@fs@mid{\kern2pt\hrule\kern2pt}%
  \let\@fs@iftopcapt\iftrue}
\floatstyle{betterruled}
\restylefloat{algorithm}
\makeatother

\makeatletter
\let\NAT@parse\undefined
\makeatother
\usepackage[hidelinks]{hyperref} 
\usepackage[capitalize, nameinlink]{cleveref}

\newtheorem{definition}{Definition}
\newtheorem{theorem}{Theorem}
\newtheorem{problem}{Problem}
\newcommand{\pathdist}{\divideontimes}
\newcommand{\TSP}{\text{TSP}}

\title{\LARGE \bf
Attrition-Aware Adaptation for Multi-Agent Patrolling
}

\author{Anthony Goeckner, Xinliang Li, Ermin Wei, and Qi Zhu
    \thanks{We acknowledge support from National Science Foundation grants 1834701, 1724341, 2038853, ECCS-2030251, and CMMI-2024774, and Office of Naval Research grant N00014-19-1-2496.}
    \thanks{All authors are with the Department of Electrical and Computer Engineering,
        Northwestern University, 2145 Sheridan Road, Evanston, IL 60208, USA. Emails:
        {\tt\small anthony.goeckner@northwestern.edu},
        {\tt\small xinliangli2023@u.northwestern.edu},
        {\tt\small ermin.wei@northwestern.edu},
        {\tt\small qzhu@northwestern.edu}.}%
}

\begin{document}

\maketitle

\begin{abstract}
    Multi-agent patrolling is a key problem in a variety of domains such as intrusion detection, area surveillance, and policing which involves repeated visits by a group of agents to specified points in an environment. While the problem is well-studied, most works do not provide performance guarantees and either do not consider agent attrition or impose significant communication requirements to enable adaptation. In this work, we present the Adaptive Heuristic-based Patrolling Algorithm, which is capable of adaptation to agent loss using minimal communication by taking advantage of Voronoi partitioning, and which meets guaranteed performance bounds. Additionally, we provide new centralized and distributed mathematical programming formulations of the patrolling problem, analyze the properties of Voronoi partitioning, and finally, show the value of our adaptive heuristic algorithm by comparison with various benchmark algorithms using physical robots and simulation based on the Robot Operating System (ROS)~2.
\end{abstract}

\section{Introduction}
\label{sec:introduction}

The multi-agent patrolling or multi-robot patrolling problem is well studied, and for good reason. The problem appears frequently in cases such as area surveillance and monitoring, police beats, intrusion detection, and even has similarities to problems such as assignment of janitorial rounds. However, agent attrition is commonly seen in practice for many of these scenarios \cite{tarantaWarfightingNeedsRobot2023} and threatens to prevent completion of the patrolling task. This attrition can take many forms, such as vehicle breakdown, robot destruction, or even a human agent calling out on sick leave.

It is important that such an attrition event can be handled in an intuitive and efficient way. For example, consider the case of human-operated vehicle patrols with the potential for vehicle breakdown. When a vehicle breakdown occurs, the other human agents must divide that agent's tasks amongst themselves and continue patrolling. Rather than reallocating all tasks to all remaining agents (the mathematically optimal solution), one natural solution is to take human behavior and limitations into account by changing the allocations of only those agents that directly neighbor the disabled agent. This serves to reduce possible human confusion and misinterpretations by limiting the amount of change required.

Such patrol algorithms should be subject to theoretical bounds, especially a bound on the loss of performance after agent attrition. Further, while many works in the patrolling literature do address agent attrition, they often do so at high communication cost, requiring a large amount of coordination among agents \cite{farinelliDistributedOnlineDynamic2017}. In this work, we devise a method for agent patrolling that is capable of operating and adapting with almost no communication, and which provides performance guarantees.

More specifically, we consider a multi-agent patrolling problem in which a team of agents must continually visit a set of observation points with a goal of minimizing the average time span that any observation point remains unvisited. To address this problem, agents must efficiently allocate observation points amongst themselves and determine an appropriate visitation order.

We first formulate this problem as an integer program and solve it repeatedly over time. The event of agent attrition may occur between problem solves. We then dynamically adjust our solution for the next solve. We note that the agent attrition event may be either stochastic, modeling unpredicted failures, or deterministic, reflecting scheduled breaks or down time with advance notice.

This problem is NP-hard, and therefore we propose a Voronoi partition-based heuristic using two-stage decomposition, where the second stage can be solved in a distributed fashion. Based around this heuristic, we devise a distributed adaptive algorithm with theoretical performance bounds.
This is the first work that we know of in the multi-agent patrolling literature to dynamically respond to agent attrition disturbances using minimal communication while limiting the amount of change required for the remaining agents and guaranteeing performance bounds.
In this paper, \textbf{we provide the following contributions}:
\begin{itemize}
    \item A centralized mathematical programming formulation (\cref{sec:formulation_centralized}) and a distributed formulation (\cref{sec:form_distrib}) for the multi-agent patrolling problem.

    \item A distributed heuristic algorithm for the multi-agent patrolling problem that can adapt to attrition of patrol agents using minimal communication while only requiring changes to a limited number of agents. (\cref{sec:heuristicApproach})

    \item A theoretical analysis of difficulties encountered when using multiplicatively-weighted Voronoi partitioning based on heterogeneous vehicle speeds. (\cref{sec:mw-voronoi})

    \item Theoretical performance bounds for our heuristic algorithm. (\cref{sec:bound-perf,,sec:bound-attrition})

    \item An experimental comparison of our heuristic algorithm with existing approaches. (\cref{sec:methods})
\end{itemize}

\section{Related Work}
\label{sec:related}
Multi-agent or multi-robot patrolling has a long history of study in the literature. Curious readers may find the latest general surveys on multi-robot patrolling \cite{basilicoRecentTrendsRobotic2022}\cite{huangSurveyMultirobotRegular2019} to be of interest. However, we focus this section on papers especially related to the present work.

A lengthy field report by Taranta et al. from the DARPA OFFSET program highlights the importance of creating a multi-agent patrolling method that minimizes communication in tactical environments and is robust to agent attrition \cite{tarantaWarfightingNeedsRobot2023}.

In \cite{schererMinMaxVertexCycle2022}, a minimum idleness connectivity-constrained multi-robot patrolling problem is introduced.
Scherer et al. partition the graph into regions but do not consider the problem of agent attrition. The primary focus of the algorithm is to maintain communication amongst all agents. A similar-looking integer program formulation of the patrolling problem is presented, though it has key differences from ours which we will describe in \cref{sec:formulation_centralized}.

Recent work by Bapat et al. \cite{bapatDistributedTaskAllocation2022} proposes two algorithms for multi-agent task allocation with highly restricted communications, one using a travelling salesman problem (TSP) heuristic. However, the algorithms are limited to one-time vehicle routing and not well-suited to the patrolling problem.

Voronoi partitioning has long been used in multi-agent task allocation, as seen in \cite{cortes2004coverage}\cite{fuLocalVoronoiDecomposition2009}\cite{dasgupta2012multi}\cite{nowzariSelftriggeredCoordinationRobotic2012}. These papers use Voronoi partitioning to divide tasks and are robust to agent failure. We expand on these with further analysis of the Voronoi partitioning with heterogeneous agent speeds and by application of Voronoi partitioning to the patrolling problem.

In a paper on the multi-depot vehicle routing problem (MDVRP), Bompadre et al. present an offline heuristic algorithm for task allocation which is the same as the initial steps of our algorithm: partition the space using the Voronoi method, and then determine the best tour through each partition using a traveling salesman problem heuristic \cite{bompadreProbabilisticAnalysisUnitDemand2007}. Alongside \cite{hwangDistributedMultiDepotRouting2014}, they also provide certain performance bounds of the algorithm which we expand on in this paper.

In work by Kim et al. \cite{kimVoronoiDiagrambasedWorkspace2020}, a task allocation system based on a Voronoi diagram for a multi-robot spray system in an orchard is proposed. This is expanded on in \cite{kimMultiplicativelyWeightedVoronoiBased2020} by use of a multiplicatively weighted (MW) Voronoi-based task allocation scheme for agricultural robots with heterogeneous speed. We expand on this paper with a discussion of some of the issues of using MW Voronoi partitioning based on vehicle speed, which were not addressed by Kim et al.

Much of the current state-of-the-art patrolling research is based on algorithms, simulators, and environments originally created by \cite{portugalDistributedMultirobotPatrol2013}\cite{portugalCooperativeMultirobotPatrol2016}. Recently, this simulator is extended by \cite{wardEmpiricalMethodBenchmarking2023} to benchmark existing state-of-the-art (SoTA) algorithms against adversarial ``attackers'', though attrition of patrol agents or other disturbances are not considered. Further experiments at the University of Bristol \cite{madinCollectiveAnomalyPerception2023} were performed using the same simulator to assess the SoTA algorithms' robustness to noisy anomaly detection at observation points.

Multi-agent reinforcement learning (MARL)-based approaches such as \cite{guoBalancingEfficiencyUnpredictability2023} are gaining popularity but to our knowledge do not as yet provide the same theoretical guarantees for performance or agent attrition as our method.

Recent work includes a method \cite{katoleBalancingPrioritiesPatrolling2023} for balancing agents' priorities amongst important and unimportant observation points which generates routes for agents while taking resource constraints into account. However, this method is not adaptive to disturbances such as agent attrition and unreliable communication. Another recent article by Kobayashi et al. presents a distributed patrol algorithm that ensures situation awareness of operators at a base station throughout the patrol \cite{kobayashiMultiRobotPatrolAlgorithm2023}, though it has the same drawbacks as \cite{katoleBalancingPrioritiesPatrolling2023}.

Indeed, many recent publications either do not address disturbances such as agent attrition or do not provide performance guarantees. Even after reviewing other articles published in the past year such as \cite{debonaAdaptiveExpectedReactive2023} and \cite{huangMultirobotCooperativePatrolling2023}, we find that no existing methods allow for adaptation to agent attrition, have minimal communication requirements, and provide theoretical performance bounds.

\section{Problem Formulation}
\label{sec:formulation}
The goal of the multi-agent patrolling problem is to repeatedly visit a set of observation points such that some metric is minimized, often based on \textit{idleness}, the amount of time each node spends between visits. Many existing works attempt to minimize the worst node idleness, the average idleness, or the standard deviation of the idleness. In our case, we attempt to minimize the \textit{average idleness} of the nodes. We now introduce the problem setting and assumptions, and then provide a brief problem statement.

\subsection{Problem Setting \& Description of Symbols}
The scalar $m$ denotes the number of agents and $n$ denotes the number of nodes. A patrol graph $G = (V, E)$ is a connected graph composed of a set of nodes $V = \{0, 1, \dots n\}$ and a set of edges $E = \{(i, j),\ i \in V\ j \in V \}$. From now on, we refer to ``observation points'' as ``nodes''. Agents are members of the set $A = \{0, 1, \dots m\}$. The binary matrix $y \in \mathbb{Z}_{m \times n}$ describes visits of agents to nodes: $y^a_{i}=1$ if and only if agent $a$ visits node $i$. The binary matrix $x \in \mathbb{Z}_{m\times n^2}$ describes whether an agent travels along an edge: $x^a_{ij}=1$ if and only if agent $a$ traverses the edge $(i, j)$. The matrix $u \in \mathbb{R}_{m \times n}$ describes the time at which an agent visits a node. The matrix $c \in \mathbb{R}_{m \times n^2}$ describes travel time between all pairs of nodes for each agent.

The vector $o \in \mathbb{Z}_n$ describes the origin node for each agent: $o^a \in \{0, 1, \dots n\}$ is the node at which agent $a$ resides at the allocation time.

We create a $m$-dimensional vector $d$ with elements $d^a \in D$. These elements are artificial \textit{destination} nodes.
The vector $d \in D_k$ denotes the destination node for each agent. One destination node $d^a$ is co-located with each origin node $o^a$ such that the travel time between them is zero: $c^a_{o^a d^a} = c^a_{d^a o^a} = 0, \forall a \in A$.

For convenience, we define the augmented set of vertices (original and artificial) as $\mathring{V} = V \cup D$.

\subsection{Assumptions}
\label{sec:assumptions}

Our assumptions are consistent with realistic patrolling scenarios. We itemize them below.
\begin{itemize}
    \item We assume that every agent $a$ begins at some node $o_a \in V$. This is a realistic assumption, especially for large patrol areas. While we do not cover selection of origin nodes in this paper, the selection does have a significant impact on algorithm performance and may be addressed by future work.

    \item Agents may become disabled at any time and cease to move. We refer to this as ``attrition''. Agents may communicate that their vehicle has been disabled. Therefore, agent attrition is immediately known to all remaining agents. We only consider attrition in terms of mobility, not communication ability.

    \item All origins are distinct. Given some origin $o^a \in o$, $$\quad o^a \neq o^{a'}, \quad\forall a' \neq a \in A.$$

    \item An agent may only travel along the edges $E$ of the graph. If $\langle i, j \rangle \notin E \enspace\exists i,j \in V$, then an agent may not directly travel from node $i$ to node $j$. An edge may be traversed by multiple agents simultaneously.

    \item Agents may be heterogeneous, each with a different speed. Agent speeds are known in advance to all agents. 
    
    \item The cost $c^a_{ij}$ is defined as travel time of agent $a$ between nodes $i$ and $j$: $$c^a_{ij} = \frac{(i \pathdist j)}  {\mathrm{speed}(a)}, \quad\forall a \in A, \quad\forall i,j \in V,$$
    where $(\cdot \pathdist \cdot)$ represents the length of the pairwise shortest path between the two nodes.
    Since the graph $G$ and agent speed are known to all agents a priori, the cost function is also known a priori. The matrix $c$ describes the all-pairs shortest path.

\end{itemize}

Having discussed the setting and assumptions above, we now present a problem statement.

\begin{problem}
    Given a patrolling graph $G = (V,E)$, set of agents $A$, and related setting and assumptions as described above, determine assignment of agents to nodes using matrix $y$ and visitation order of assigned nodes using matrix $x$ in order to minimize the average time between visits (idleness time) of nodes in the patrol graph.
\end{problem}

\section{Mathematical Programming Approach}
\label{sec:formulation_centralized}
We initially based our formulation on the deterministic single-agent Profitable Tour Problem (PTP) as described in \cite{archettiChapter10Vehicle2014}\cite{bruniRiskaverseProfitableTour2019}, which has an objective of minimizing the travel cost minus profits collected at each node. However, we realized that the problem could be simplified by assuming that each agent repeatedly visits the same nodes in a closed cycle and by attempting to minimize the average length of these cycles.
Our formulation is most accurately categorized as a multi-depot vehicle routing problem (MDVRP), since we find a cycle beginning at a pre-selected origin point for each agent.
This makes our formulation more similar to that discussed in \cite{schererMinMaxVertexCycle2022}, which seeks to minimize the maximum amount of time that any node is left idle (between visits). However, our objective is to minimize the \textit{average} time that all nodes are left idle. We formulate this as a mixed-integer nonlinear program for readability, although there may be methods available to linearize the problem.

Based on the problem formulation above, we develop our multi-agent patrolling approach as follows. Note that we use a ``big-M'' formulation for constraint linearization, with some sufficiently large constant $M$ that is significantly greater than any other variables.

\begin{subequations}\label{eq:newform}
    \begin{align}
        \min_{x,y,u} & z = \frac{1}{n} \sum_{a=0}^m (u^a_{d_a} \sum_{i=0}^n y^a_i) & \tag{\ref{eq:newform}} \\
        \mbox{s.t.}
        & \sum_{\substack{i=0 \\ i\neq j}}^n x^a_{ij} = y^a_{j} \quad& \forall a \in A, \forall j \in \mathring{V} \label{eq:newform_constr_arc_prev} \\
        & \sum_{\substack{j=0 \\ j\neq i}}^n x^a_{ij} = y^a_{i} \quad& \forall a \in A, \forall i \in \mathring{V} \label{eq:newform_constr_arc_next}  \\
        & u^a_{i} + c^a_{ij} - M(1 - x^a_{ij}) \leq u^a_{j} & \forall a \in A, \forall i \in V, \label{eq:newform_constr_no_subtour}\\
            &&\forall j \in \mathring{V} \setminus \{i\} \nonumber \\
        & u^a_{o_a} = 0 \quad& \forall a \in A \label{eq:newform_constr_origin_time} \\ 
        & x^a_{d^a o^a} = 1 \quad& \forall a \in A \label{eq:newform_constr_cycle} \\
        & \sum_{i \in D \setminus \{d^a\}} y^a_i = 0 \quad& \forall a \in A \label{eq:newform_constr_wrongdest} \\
        & \sum_{a=0}^m y^a_{i} \geq 1 \quad& \forall i \in V \label{eq:newform_constr_visit_all} \\
        & x^a_{ij} \in \{0, 1\} \quad& \forall a \in A, \forall i,j \in \mathring{V} \nonumber \\
        & y^a_{i} \in \{0, 1\} \quad& \forall a \in A, \forall i \in \mathring{V} \nonumber \\
        & u^a_{i} \geq 0 \quad& \forall a \in A, \forall i \in \mathring{V} \nonumber
    \end{align}
\end{subequations}

This formulation results in each agent being assigned an ordered series of nodes to visit, beginning and ending at the same location (since $o^a$ and $d^a$ are co-located $\forall a$). This creates a closed patrol cycle for each agent.

\paragraph{Constraints}
Constraint \eqref{eq:newform_constr_arc_prev} enforces that if agent $a$ visits node $j$, $a$ must have traversed some edge connected to $j$ (a ``come-from'' constraint). Constraint \eqref{eq:newform_constr_arc_next} enforces that agent $a$ may only traverse a single edge immediately succeeding node $i$ (a ``go-to'' constraint).
Constraint \eqref{eq:newform_constr_no_subtour} enforces that there are no subtours and records the visit time of agent $a$ to node $j$ as the visit time at node $i$ plus the travel cost (time) between nodes $i$ and $j$. This is similar to the ``subtour elimination constraint'' from the Miller-Tucker-Zemlin formulation of the TSP \cite{millerIntegerProgrammingFormulation1960}.
Constraint \eqref{eq:newform_constr_origin_time} enforces that the origin node is visited at time $0$.
Constraint \eqref{eq:newform_constr_cycle} enforces that an agent $a$ must travel from its artificial destination node $d^a$ to its co-located origin node $o^a$, creating a cycle.
Constraint \eqref{eq:newform_constr_wrongdest} enforces that an agent not visit the artificial destination nodes of other agents. Constraint \eqref{eq:newform_constr_visit_all} enforces that all nodes must be visited.

\paragraph{Objective Function}
The objective function \eqref{eq:newform} describes the \textit{average idleness time} of all nodes. For each agent $a \in A$, we multiply the total time $u^a_{d^a}$ taken to complete a cycle by the number of nodes $\sum_{i \in V} y^a_i$ visited in that cycle. We define \textit{idleness} as the time span between agent visits. Each of these nodes will then have the same idleness, since only one agent is assigned to each cycle. Summing the results for each agent $a$ provides the total idleness time of all nodes, which we divide by $n$ to find the average idleness time. Node idleness time is a commonly used metric in the literature for comparison of patrolling algorithms, as seen in e.g. \cite{portugalDistributedMultirobotPatrol2013}\cite{farinelliDistributedOnlineDynamic2017}\cite{bruniRiskaverseProfitableTour2019}. By using this objective function, we can more easily compare our solution with existing approaches.

Note an important difference here between our formulation and that of \cite{schererMinMaxVertexCycle2022}'s min-max vertex cycle cover problem. While Scherer et al. attempt to minimize the worst idleness time of any node, we attempt to minimize the average idleness time. Further, by use of the assignment decision matrix $y$, our problem may be more easily decomposed to a distributed problem (see \cref{sec:form_distrib}).

\subsection{Distributed Approach}
\label{sec:form_distrib}
In multi-agent systems where attrition is expected to occur or in systems that must be robust to disturbances, such as those studied in \cite{tarantaWarfightingNeedsRobot2023}, a single point of failure is often unacceptable. Therefore, we also devise a distributed approach to the problem which we will later use as the basis for our adaptive heuristic algorithm in \cref{sec:heuristicApproach}.

To create a two-stage distributed version of our problem, we break it into master and sub-problems. The master problem should allocate nodes to agents, and the sub-problem should determine the visitation order of the allocated nodes. Once the master problem generates an assignment, the sub-problems can be solved using local information in a distributed and parallel fashion.
We observe that $y$ (node allocation) is a confounding variable, which when fixed the problem naturally decomposes into $m$ sub-problems, one for each agent. As described by \cite{morVehicleRoutingProblems2022}, each of these sub-problems is a traveling salesman problem (TSP).

We define a sub-problem for each agent $a$ as follows, where $\hat{y}^a$ is the fixed $y^a$ chosen by the master problem:
\begin{subequations}
    \label{eq:sub}
    \begin{align}
        \min_{x^a,u^a} & z^a = \frac{u^a_{d_a}}{n} \sum_{i=0}^n \hat{y}^a_i & \tag{\ref{eq:sub}} \\
        \mbox{s.t.}
        & \sum_{\substack{i=0 \\ i\neq j}}^n x^a_{ij} = \hat{y}^a_{j} \quad& \forall j \in \mathring{V} \label{eq:sub_constr_arc_prev} \\
        & \sum_{\substack{j=0 \\ j\neq i}}^n x^a_{ij} = \hat{y}^a_{i} \quad& \forall i \in \mathring{V}\label{eq:sub_constr_arc_next}  \\
        & u^a_{i} + c^a_{ij} - M(1 - x^a_{ij}) \leq u^a_{j} \quad& \forall i \in V, \label{eq:sub_constr_no_subtour}\\
            &&\forall j \in \mathring{V} \setminus \{i\} \nonumber \\
        & u^a_{o_a} = 0 \quad& \label{eq:sub_constr_origin_time} \\ 
        & x^a_{d^a o^a} = 1 \quad& \label{eq:sub_constr_cycle} \\
        & x^a_{ij} \in \{0, 1\} \quad& \forall i,j \in \mathring{V} \nonumber \\
        & u^a_{i} \geq 0 \quad& \forall i \in \mathring{V} \nonumber
    \end{align}
\end{subequations}

Given a predefined set of nodes to visit $\{i | \hat{y}^a_i = 1\}$ by the master problem, this sub-problem finds the optimal tour for agent $a$ to visit all nodes in the set. This is clearly a TSP.

\subsection{Agent Attrition}
When an agent $a$ suffers attrition during execution, nodes assigned to that agent must be reassigned to other agents. In terms of the original MIP described in \cref{sec:formulation_centralized}, this is equivalent to adding constraints $y^a_i=0 \quad\forall i \in \mathring{V}$ and modifying \eqref{eq:newform_constr_cycle} to have a right-hand side of $0$.

Performing this constraint modification is impractical for real-time adaptation, requiring a re-solve of the MDVRP, which is NP-Hard. Therefore, we introduce a heuristic approach to the problem to allow for real-time adaptation.

\section{Heuristic Approach}
\label{sec:heuristicApproach}
The nonlinear optimization problem described in \cref{sec:formulation} is NP-Hard and thus is not suitable for computation at runtime or with large numbers of agents or nodes.

As seen in \cref{sec:form_distrib}, the problem may be divided into two phases: node assignment phase and visitation ordering phase. While these phases are somewhat interdependent, we attempt to create a heuristic for each of the two phases.
Then we describe an adaptive algorithm to solve the problem of patrolling in the face of agent attrition.

\paragraph{Node Assignment Heuristic}
As a heuristic for the node assignment phase, we use Voronoi partitioning of nodes based on the agent's origin point and using travel time as the distance measure. For each node, we calculate its shortest wait time to each of the origins and assign it to the best one. Formally, for each node $i$, we first define an optimal agent selection function $a^*(i)$ by solving the problem of 
\begin{equation} a^*(i) \in \arg\min_{a\in A} c^a_{io^a} = (i \pathdist o^a) / \mathrm{speed}(a)\label{eq:optimal_agent}\end{equation}
If there are multiple minimizers for this problem, we randomly pick one to assign to $a^*(i)$. Then we assign the decision variable 
\begin{align}
 \hat y^a_i = 
    \quad
    \begin{cases}
    1 \quad \mbox{if $a=a^*(i)$},\\ 
    0 \quad \mbox{otherwise.} 
    \end{cases}
    \quad\label{eq:voronoi_y}
\end{align}
We select this simple geometric heuristic for good reason: it allows for the loss of an agent and subsequent reallocation of that agent's nodes without reallocation of the entire graph. See \cref{sec:attrition} for more information and proof.

\paragraph{Visitation Order Heuristic}
The visitation ordering phase is a TSP \cite{morVehicleRoutingProblems2022}. We use a simple nearest-neighbor heuristic to find an approximate solution. This is a greedy algorithm which, starting at the origin, selects the nearest node that is yet to be visited, then performs the same operation at the selected node, and so on until all nodes have been visited.

\subsection{The Adaptive Heuristic Algorithm}
\label{sec:algorithm}
We create a runtime heuristic algorithm that enables patrolling while dynamically responding to agent attrition (due to vehicle breakdown, etc.) by reallocating select agents to cover for the lost agent. The attrition and adaptation process is covered in detail below in \cref{sec:attrition}.

Before patrolling can begin, all nodes in the graph must be allocated to agents using Voronoi partitioning based on the agent's starting position.

\begin{algorithm}
    \caption{The Adaptive Heuristic-based Patrolling Alg.}
    \begin{algorithmic}[1]
        \State $y \gets$ voronoiPartition($G$) \label{lst:line:initialStart}
        \State $x^a, u^a \gets$ \TSP($y^a$)
        \State Begin patrolling.  \label{lst:line:initialStop}
        \While{patrolling}
        \If{another agent $\tilde a$ suffers attrition}
            \State $y' \gets$ voronoiPartition($G$) 
            \If{$y'^a \neq y^a$}  \Comment{Did our alloc. change?}
                \State $x^a, u^a \gets$ \TSP($y'^a$) \Comment{Yes, do TSP again.}
            \EndIf
            \State $y \gets y'$
        \EndIf
        \EndWhile
    \end{algorithmic}
    \label{alg:ahpa}
\end{algorithm}

We refer to this algorithm as the Adaptive Heuristic-based Patrolling Algorithm (AHPA), and it is labelled as ``AHPA'' in any figures or graphs.

AHPA runs concurrently on all agents. The only inter-agent communication required by our algorithm is a notification of agent attrition. While we currently model this as an explicitly communicated message, it could just as easily be an observation or other non-networked indication. The potential use of observations is aided by the fact that only neighboring agents must change their allocation, as will be shown in \cref{thm:neighborAllocations}.

\subsubsection{Agent Attrition \& Adaptation}
\label{sec:attrition}
When agent attrition occurs due to vehicle breakdown, the agent's observation points must instead be visited by other agents. As seen in \cref{alg:ahpa}, we perform Voronoi partitioning of the entire graph $G$, resulting in a new assignment matrix $y'$. We feed this new assignment matrix into our visitation order heuristic to calculate routes for the remaining agents.

One of our design goals is to minimize disruption to remaining agents upon agent loss. Therefore, we use Voronoi partitioning for the assignment heuristic and consider the case where all vehicle speeds are the same. In case of agent loss, only the assignments of the agent's immediate Voronoi neighbors (see \cref{def:neighbor}) will change.

\begin{definition}[Neighboring Agent]\label{def:neighbor}
    The agent $b \in A$ is considered a neighbor of agent $a \in A$ if for all nodes $i$ with $a^*(i)=a$ and $j$ with $a^*(j)=b$ and for all $t \in \mbox{path}(i,j)$, $a^*(t) = a$ or $a^*(t)=b$, where $\mbox{path}(i,j)$ refers to any shortest path between $i$ and $j$.
\end{definition}



\begin{theorem}\label{thm:neighborAllocations}
    The loss of a single agent $\tilde a$ will only change the allocations of its neighboring agents for the heuristics described in \eqref{eq:voronoi_y} when all agents move at the same speed.
\end{theorem}
\begin{proof}
    For any nodes $i$ with $a^*(i)\neq \tilde a$, the original optimal agent assignment remains optimal and therefore we only change the allocation for the nodes with $a^*(i) = \tilde a$. We proceed to prove by contradiction. Suppose for such node $i$, 
    \begin{equation} a^*(i) \in \arg\min_{a\in A-\{\tilde a\}} c^a_{io^a} =c,\label{eq:skip_neighbor}\end{equation} where $c$ is not a neighboring agent of $\tilde a$. This implies that the path between $i$ and $c$ goes through a node $j$, with $a^*(j) = b$ and $b$ is a neighboring agent  of $\tilde a$. For this node $j$, we also have that \begin{equation}(j \pathdist b) < (j \pathdist c) \label{ineq:dist}, \end{equation} and therefore it is in the interior of the partition associated with agent $b$. If such $j$ does not exist, then agents $a$ and $c$ would be neighbors by \cref{def:neighbor}. 

    By \eqref{eq:skip_neighbor}, we have 
    \[c_{ic}^c\leq c_{ib}^b.\] Furthermore, since all agents' speeds are the same, we have 
    \[(i \pathdist j)+(j \pathdist c) = (i \pathdist c)\leq (i \pathdist b).\]
Since the definition of $(i \pathdist b)$ calculates the shortest path distance between $i$ and $b$, we have 
\[(i \pathdist b)\leq (i \pathdist j) + (j \pathdist b).\] The previous two relations imply that
\[(j \pathdist c)\leq (j \pathdist b),\] which contradicts \eqref{ineq:dist} and completes the proof. 
\end{proof}

\subsubsection{The Case of Heterogeneous Speed}
\label{sec:mw-voronoi}
We note that \cref{thm:neighborAllocations} does not hold when the agents have different speeds. In fact, multiplicatively-weighted Voronoi partitioning based on different agent speeds results in unexpected nonconvex/non-continuous allocation results. Consider the following simple case, where many nodes are allocated to agents on the real line with three agents placed at locations $0,1,2,$ respectively, from left to right. We refer to these agents by their locations. The agent at $2$ has a speed of $2$ units/s and agents at $0, 1$ both move at a speed of $1$ unit/s. Therefore, any node $i$ which lies in $(-\infty, -2]$ has $a^*(i)=2$, any node $i$ in  $(-2, 1/2]$ has $a^*(i) = 0$, any node $i$ in $(1/2, 4/3]$ has $a^*(i) = 1$, and any node $i$ in $(4/3, \infty)$ has $a^*(i)=2$.  Note that due to the different speeds, even before any agent attrition, the partitions are not contiguous, as shown in \cref{fig:voronoi_counter}. Namely, agent $2$ covers the area nearby and also area far enough to the left of agent $0$, since its speed is the highest. If agent $0$ breaks down, then the point $-1$ would take time $2$s to reach for the agent at $1$ and $3/2$s for the agent at $2$, and therefore will be assigned to the agent at location $2$, who is not a neighbor of agent $0$, thus invalidating \cref{thm:neighborAllocations}.
This aspect of multiplicatively-weighted Voronoi partitioning is not addressed by earlier works such as \cite{kimMultiplicativelyWeightedVoronoiBased2020}. While our \cref{thm:neighborAllocations} is unable to handle this case, we wish to present it for further discussion.

\begin{figure}[ht]
    \centering
    \includegraphics[scale=0.3]{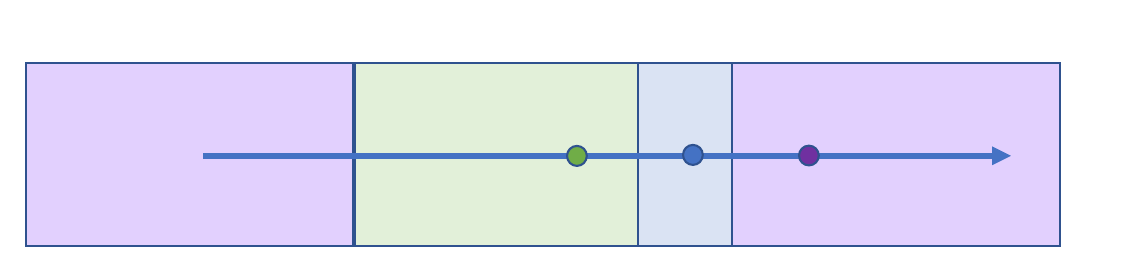}\caption{Discontinuous weighted Voronoi partition based on heterogeneous agent speed, illustrating the dilemma described in \cref{sec:mw-voronoi}. Dots represent agents and colored boxes represent partitions assigned to the agent of the same color. The right-most agent (purple) has significantly higher speed than the other two, resulting in unexpected partitioning results.}\label{fig:voronoi_counter}
\end{figure}


\subsection{AHPA Performance Bound}
\label{sec:bound-perf}
At its core, AHPA uses Voronoi partitioning to divide the environment and then determines the best tour of each partition using a TSP heuristic. Bompadre et al. derive an $\Omega(m)$ lower bound for the Voronoi partitioning with TSP method in the context of the multi-depot vehicle routing problem, which describes it as an $m$-approximation process, where $m$ is the number of depots \cite{bompadreProbabilisticAnalysisUnitDemand2007}. In our case, we view depots as agent origins (depots serving a single agent), and AHPA can be reduced to the same algorithm as \cite{bompadreProbabilisticAnalysisUnitDemand2007} by considering only the initial operations at lines~\ref{lst:line:initialStart}-\ref{lst:line:initialStop} of \cref{alg:ahpa}. Therefore, AHPA is an $m$-approximation of the optimization problem presented in \cref{sec:formulation_centralized}.

\subsection{Bound on Performance Loss after Attrition}
\label{sec:bound-attrition}
One of AHPA's main goals is robustness to agent attrition with little communication overhead. In this section, we provide a bound on the algorithm's maximum performance loss after agent attrition occurs.

Thanks to prior works, we know that the approximation ratio of AHPA to the optimal solution is $\Omega(m)$ \cite{bompadreProbabilisticAnalysisUnitDemand2007} and $O(m)$ \cite{hwangDistributedMultiDepotRouting2014}. Therefore, the tightest approximation ratio of AHPA is $\Theta(m)$. The bound on performance loss after agent attrition is the ratio of the original approximation ratio for $m$ agents and the new approximation ratio for $m-1$ agents after attrition:
$$\frac{\Theta(m)}{\Theta(m - 1)}$$

Hence, attrition of one agent roughly translates to $\Theta(\frac{m}{m-1})$ increase in average node idleness time.

\section{Experimentation \& Results}
\label{sec:methods}
We experimentally compare our approach to several existing state-of-the-art dynamic patrolling methods using both physical robots and a realistic simulation environment. Both physical and simulated experiments use the \textit{Grex} framework, which we developed for general MAS research.

\subsection{Experimental Environment}
Much prior work on patrolling algorithms, including that by \cite{portugalDistributedMultirobotPatrol2013}\cite{portugalCooperativeMultirobotPatrol2016}\cite{farinelliDistributedOnlineDynamic2017}\cite{wiandtSelforganizedGraphPartitioning2017}\cite{elgibreenDynamicTaskAllocation2019}, uses a patrolling simulator originally developed by David Portugal et al. in \cite{portugalDistributedMultirobotPatrol2013} and further papers. Many of the patrolling algorithms developed using this simulator are still considered to be the state-of-the-art. 

Separately, we have developed a multi-agent simulation and experimentation framework based on ROS~2, the Robot Operating System, which may be seen in \cref{fig:mas-framework}. Our framework, Grex, is capable of execution on real robots or may be used with a variety of compatible simulators. Using this multi-agent framework, we are able to write agent and experiment code which is common both to physical platforms and to a variety of simulators. ROS~2 provides us with features that are shared across these physical and simulated platforms including data visualization, communications introspection, runtime configuration, and a wide-ranging library of third-party modules for easy integration of new capabilities.

For the purposes of this paper, we leverage that extensibility to integrate the environments and patrolling algorithms created by Portugal et al. \cite{portugalDistributedMultirobotPatrol2013} and others with our framework\footnote{The integrated experiment code is publicly available at \url{https://github.com/NU-IDEAS-Lab/patrolling_sim}}. This greatly simplifies comparison with existing algorithms.

\begin{figure}[ht]
    \centering
    \includegraphics[width=0.43\textwidth]{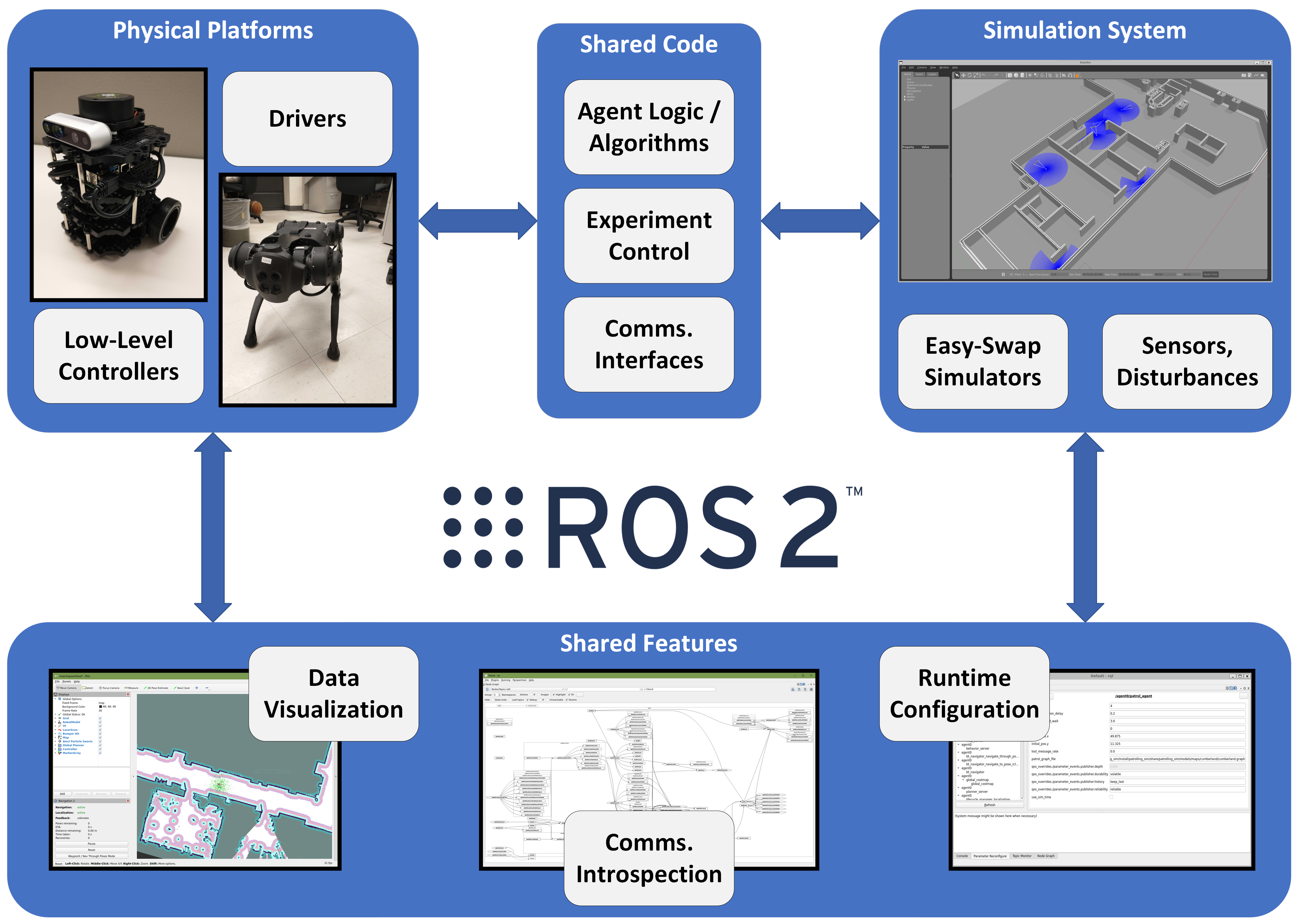}
    \caption{We use the MAS framework, Grex, for straightforward experimentation in both simulated environments and with physical robots. At top-left are the TurtleBot3 robots which we use in physical experiments. At top-right, simulated robots patrol the ``Cumberland'' environment from \cite{portugalDistributedMultirobotPatrol2013}.}
    \label{fig:mas-framework}
\end{figure}

We simulate sensor noise based on Gaussian distributions; this is the primary source of uncertainty in the simulation. Each agent performs its own localization and navigation. Agents may collide or interference with each other, resulting in delay reaching goals. For physical robots, real disturbances are at play, including sensor noise, localization uncertainty, communication losses, and robot collisions.

\subsection{Benchmark Algorithms}
\label{sec:benchmark_algorithms}

We select four benchmark algorithms to compare against our own solution.
Two are greedy algorithms that serve as a baseline: Greedy Bayesian Strategy (GBS) \cite{portugalDistributedMultirobotPatrol2013} and DTA-Greedy \cite{farinelliDistributedOnlineDynamic2017}.
We also compare against the more sophisticated Concurrent Bayesian Learning Strategy (CBLS) \cite{portugalCooperativeMultirobotPatrol2016} and DTA-Partitioning (DTAP) \cite{farinelliDistributedOnlineDynamic2017}.


\subsection{Experimental Procedure}
\label{sec:procedure}
For each of the algorithms described in \cref{sec:benchmark_algorithms} plus our own \cref{alg:ahpa} (AHPA), we perform tests using six agents in both the 40-node ``Cumberland'' environment from \cite{portugalDistributedMultirobotPatrol2013} and our own 18-node ``L440'' environment, an empty conference room. A ``Cumberland'' test lasts thirty minutes (1800 seconds). Each algorithm is tested twice: once with no agent attrition and once with attrition of single randomly-selected agents at the 300-second and 1300-second marks. A test in the ``L440'' environment lasts for five minutes.

For each algorithm and each test, we execute three runs over which results are averaged to account for possible  discrepancies. The test system is automated, and experimental monitoring and control is performed identically for each algorithm. All simulation tests were performed on the same machine.

\subsection{Results}
In analysis of the results, our primary focus is on the ``average idleness'' metric which we identify as our objective function in \eqref{eq:newform}. We look at this value over time to evaluate the performance of our algorithm.

Overall, our AHPA algorithm performs well in comparison with existing approaches and effectively addresses the problem statement in \cref{sec:formulation}. As seen in the left-hand column of \cref{fig:cumberland}, the AHPA algorithm outperforms all others in non-attrition trials. We attribute this to AHPA's partitioning of the environment and surmise via observation of experiments that many of the other algorithms suffer from either poor/overlapping partitioning (DTAP) or from physical collisions amongst the agents (DTAG, CBLS), which significantly impacts performance. AHPA's solution results in non-overlapping partitions which improves allocations while also decreasing physical interference and collisions amongst the agents. As expected, use of AHPA results in a lower standard deviation of node idleness than other approaches, since each agent is assigned a tour which they then patrol continually. This may be seen in the left-hand column of \cref{fig:cumberland}.

\begin{figure}[ht]
    \centering
    \includegraphics[width=0.45\textwidth]{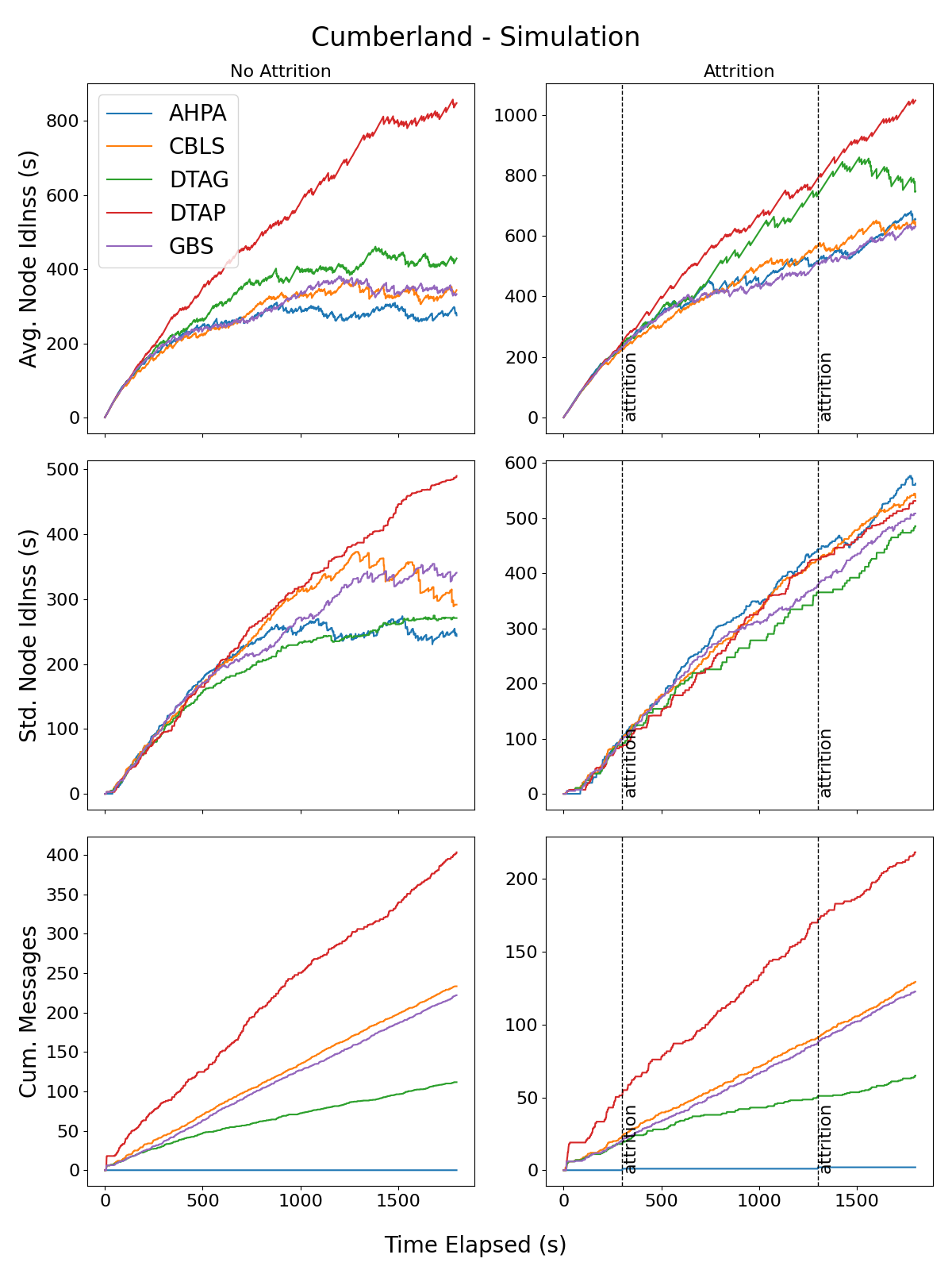}
    \caption{Performance of the algorithms over time in the simulated Cumberland environment originally used in \cite{portugalDistributedMultirobotPatrol2013}. At left, results with no agent attrition. Note the excellent comparative performance of AHPA (in blue). At right, with two instances of agent attrition. Note the cumulative message difference between AHPA and the other algorithms.}
    \label{fig:cumberland}
\end{figure}

In the attrition test, we remove a randomly-selected agent from the simulation at 300 seconds and another agent at 1300 seconds. The results of this are visible in the right-hand column of \cref{fig:cumberland}. While AHPA does not recover faster than other algorithms, it requires far fewer inter-agent messages than the other algorithms to do so. As seen in \cref{fig:cumberland}, it only requires a single message in case of agent attrition. Contrast this with the many messages used by other algorithms, which can be problematic in communication-constrained environments.

\begin{wrapfigure}{L}{0.215\textwidth}
    \centering
    \includegraphics[width=0.225\textwidth]{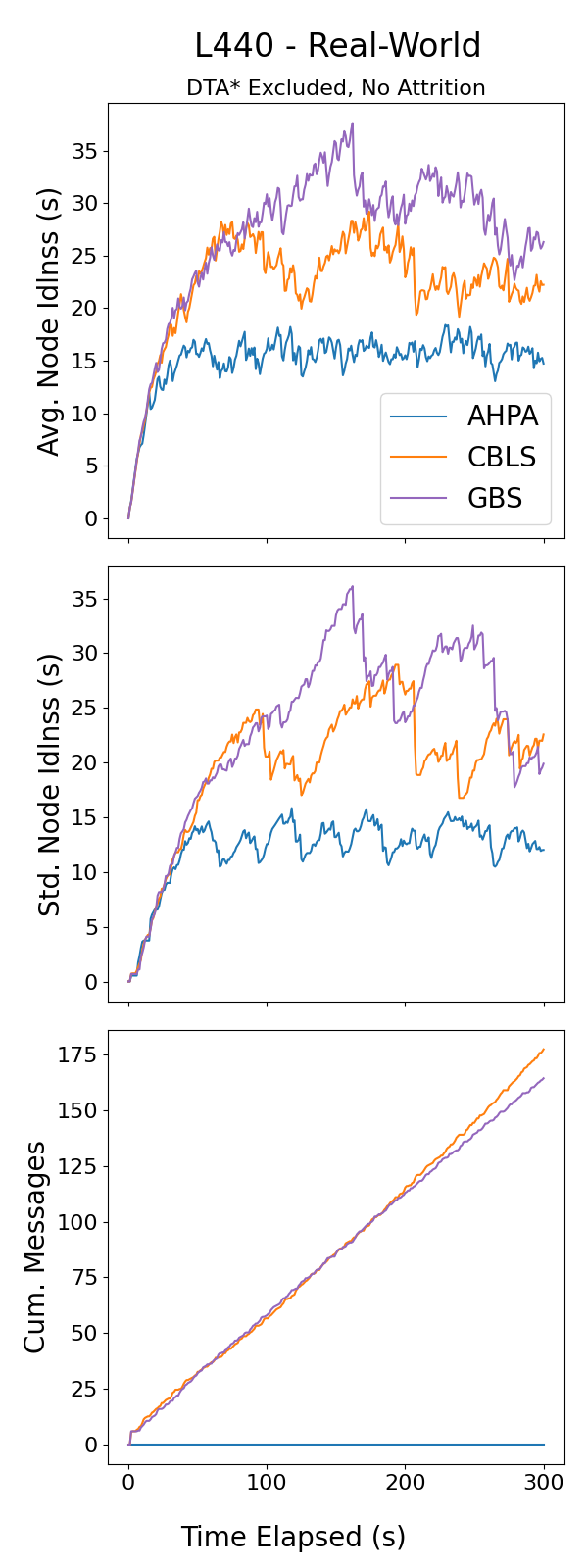}
    \caption{Performance of AHPA and selected benchmark algorithms in real experiments. Note the stability of AHPA throughout. The DTAP and DTAG algorithms are not shown, as their performance was so poor that the graph became difficult to read.}
    \label{fig:L440}
\end{wrapfigure}

Physical experiments demonstrate results consistent with those from simulation. AHPA exhibits stable operation throughout the test and outperforms the benchmark algorithms. However, we are surprised by the poor performance of the DTAP and DTAG algorithms. We cannot reproduce the high performance of those algorithms seen in \cite{farinelliDistributedOnlineDynamic2017}. In fact, DTAP and DTAG perform so poorly in our physical tests, with average idleness times nearing $252$ and $216$ seconds, respectively, that we exclude them from \cref{fig:L440} for readability. We observe that their poor performance is due to physical collisions and interference amongst agents during mission execution. However, other algorithms do not suffer from this problem to the same extent.

\paragraph*{Performance Bounds}
To demonstrate the AHPA performance bound described in \cref{sec:bound-perf}, we also compare AHPA with the optimal solution (computed using a commercial solver) for a toy problem with $n=6$ nodes and $m=2 \dots 6$ agents. The results may be seen in \cref{fig:bounds}, and demonstrate that useful performance bounds are provided.

\begin{wrapfigure}{R}{0.25\textwidth}
    \centering
    \includegraphics[width=0.25\textwidth]{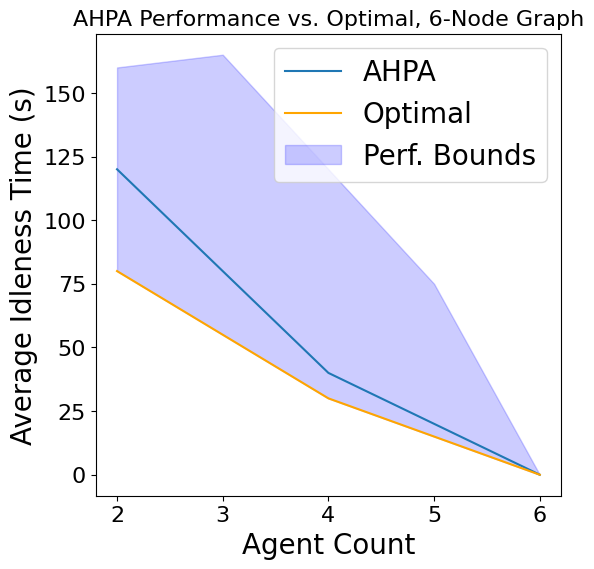}
    \caption{Performance of AHPA compared to the optimal performance. Note that AHPA remains within the performance bound described in \cref{sec:bound-perf}.}
    \label{fig:bounds}
\end{wrapfigure}

\section{Conclusion}
In this paper, we analyze the multi-agent patrolling problem and make a number of contributions to it. Importantly, we solve it with an adaptive heuristic algorithm (AHPA), using Voronoi partitioning to allocate nodes to agents and a TSP heuristic to determine visitation order for each agent. We provide performance guarantees for this algorithm, including for the case of agent attrition. This is the first patrolling algorithm that we know of to provide such guarantees.

Further, we devise new centralized and decentralized integer programming formulations for the patrolling problem which we feel are better-suited than existing MDVRP formulations to many patrolling scenarios.

We show that in the case of heterogeneous agent speeds, the behavior of Voronoi partitioning based on speed yields unexpected results. This has not been reported before to our knowledge, even though existing works use speed-weighted Voronoi partitioning for allocation.

We also show that in the case of homogeneous agent speeds, the attrition of a single agent will only affect its neighboring agents' allocations, which is beneficial in certain patrolling scenarios, especially when humans are involved.
By only changing the allocations of neighboring agents after attrition, we hope that future research will be able to devise a communication-free patrolling methodology relying only on local agent observations of neighbors to perform adaptation in the face of attrition.

Together, the above contributions form a highly effective multi-agent patrolling algorithm, AHPA, which is capable of adapting to agent attrition with minimal communication requirements and is subject to guaranteed performance bounds. As seen in \cite{tarantaWarfightingNeedsRobot2023}, this work is highly valuable in practice and we hope that it will lead to additional research focus on these practical considerations of attrition and disturbances in MAS.


\section*{ACKNOWLEDGMENTS}
Many thanks to Mr. Yixuan Wang, Mr. Henry Abrahamson, Dr. Xiangguo Liu, and Dr. Karen Smilowitz for their feedback and advice during the development of this paper.

\bibliographystyle{IEEEtran}
\bibliography{IEEEabrv,refs}

\end{document}